\documentclass[12pt]{article}
\usepackage{amsmath,amssymb,amsthm,latexsym,euscript,amscd, authblk}
\usepackage[english]{babel}
\usepackage[all]{xy}
\usepackage[dvips]{graphics}

\newtheorem{Theorem}{Theorem}
\newtheorem{Proposition}[Theorem]{Proposition}
\newtheorem{Corollary}[Theorem]{Corollary}

\title{On the noncommutative deformation of the operator graph corresponding to the Klein group}

\author[1]{G.G. Amosov\thanks {gramos@mi.ras.ru}}

\author[2, 3]{I.Yu. Zhdanovskiy\thanks {ijdanov@mail.ru}}

\affil[1]{Steklov Mathematical Institute of Russian Academy of Sciences}
\affil[2]{Moscow Institute of Physics and Technology}
\affil[3]{National Research University High School of Economics, Laboratory of Algebraic Geometry}

\begin{document}

\maketitle

\begin{abstract}
We study the noncommutative operator graph ${\mathcal L}_{\theta }$ depending on complex parameter $\theta $ recently introduced by M.E. Shirokov to construct channels with positive
quantum zero-error capacity having vanishing n-shot capacity. We define the noncommutative group $G$ and the algebra ${\mathcal A}_{\theta }$ which is a quotient of ${\mathbb C}G$
with respect to the special algebraic relation depending on $\theta $ such that the matrix representation $\phi $ of ${\mathcal A}_{\theta }$ results in the algebra ${\mathcal M}_{\theta }$ generated by ${\mathcal L}_{\theta }$. In the case of $\theta =\pm 1$ $\phi $ is degenerated to the faithful representation of ${\mathbb C}K_4$, where $K_4$ is the Klein group. Thus, ${\mathcal L}_{\theta }$ can be considered as a noncommutative deformation of the graph associated with the Klein group.
\end{abstract}

\section{Introduction.}

Denote $\mathfrak {B}({\mathcal H})$, $\mathfrak {T}({\mathcal H})$ and $\mathfrak {S}({\mathcal H})$ the algebra of all bounded operators, the Banach space of all trace class operators and the convex set of
positive unit-trace operators (quantum states) in a Hilbert space $\mathcal H$, respectively.
A completely positive trace-preserving linear map $\Phi :\mathfrak {T}({\mathcal H})\to \mathfrak {T}({\mathcal H})$ is said to be a quantum channel.
Given a quantum channel $\Phi $
one can pick up its Kraus decomposition of the form
\begin{equation}\label{Kraus}
\Phi (\rho)=\sum \limits _{n}V_n\rho V_n^*,
\end{equation}
where $V_n\in \mathfrak {B}({\mathcal H}),\ \sum \limits _nV_n^*V_n=I$, $I$ is the identity operator in $\mathcal H$, $\rho \in \mathfrak {S}({\mathcal H})$ \cite{hol}. Throughout all this paper the Hilbert space $\mathcal H$ is supposed to be finite-dimensional.

In \cite{graph} the noncommutative graph $\mathcal{G}(\Phi)$ of a quantum channel $\Phi $ was introduced. Let us consider
the Kraus decomposition (\ref {Kraus}), then $\mathcal {G}(\Phi)$ is the operator subspace
\begin{equation}\label{graph}
\mathcal {G}(\Phi)=\overline {Lin(V_j^*V_k)}.
\end{equation}
In \cite{graph2, graph3} it was shown that the operator subspace $\mathcal S$ is associated with some channel in the sense of (\ref {graph}) iff $I\in \mathcal S$
and $\mathcal S^*=\mathcal S$. It seems to be interesting to study the algebraic structure of (\ref {graph}). For example, the question arises whether it is possible to
show that $\mathcal {G}(\Phi )$ is an image of the representation of the algebra ${\mathbb C}G$ associated with some noncommutative group $G$? Here we shall consider
this question for an important particular case.

Consider the operator subspace ${\cal L}_{\theta} \subset Mat_4({\mathbb C})$:
\begin{equation}\label{L}
\begin{pmatrix}
a & b & c\theta & d\\
b & a & d & c/\theta\\
c/\theta & d & a & b\\
d & c\theta & b & a
\end{pmatrix}
\end{equation}
for $\theta \in {\mathbb C}^*$.
The operator graph (\ref {L}) was introduced in \cite {maxim} to construct channels with positive
quantum zero-error capacity having vanishing n-shot capacity.

A finite-dimensional channel $\Phi :\mathfrak {S}({\mathcal H})\to \mathfrak {S}({\mathcal H})$ is called pseudo-diagonal \cite{ruskai} if
\begin{equation}\label{pseudo}
\Phi (\rho )=\sum \limits _{i,j}c_{ij}<\psi _i|\rho |\psi _j>|i><j|
\end{equation}
where $\{c_{ij}\}$ is a Gram matrix of a collection of unit vectors, $\{|\psi_i>\}$ is a collection of vectors in $\mathcal H$ such that $\sum \limits _i|\psi _i><\psi _i|=I$
and $\{|i>\}$ is an orthonormal basis in $\mathcal H$. Pseudo-diagonal channels are complementary to entanglement-breaking channels and vice versa.
In \cite{maxim} it was shown that the noncommutative graph ${\cal L}_{\theta }$ can be associated with a family of pseudo-diagonal channels (\ref {pseudo})
depending on the parameter $\theta $.

In the present paper we shall show that ${\cal L}_{\theta }$ can be considered as a image of representation $\pi _{\theta }$ of the ring generated by the
group $G$ with three generators $x,y,z$ satisfying the relations
\begin{equation}\label{relations}
x^2=y^2=z^2=1,\ xz=zx,\ yz=zy.
\end{equation}
Notice that adding to (\ref {relations}) the relation
\begin{equation}\label{klein}
xy=yx=z
\end{equation}
we obtain the Klein group $K_4$.

The Klein group $K_4=\{1,x,y,z\}$ with the generators satisfying (\ref {relations})-(\ref {klein}) is abelian. Thus, all its irreducible representations are one-dimensional.
It implies that the minimal dimension of any faithful representation is equal to four. Pick up the orthonormal basis $\{|j>,\ 1\le j\le 4\}$ in the Hilbert space ${\mathcal H}_4,\ dim{\mathcal H}_4=4$.
Then, we can define the standard faithful representation of $K_4$ in ${\mathcal H}_4$ by the formula
$$
x|1>=y|1>=z|1>=|1>,
$$
\begin{equation}\label{kanal}
x|2>=y|2>=-|2>,\ z|2>=|2>,
\end{equation}
$$
x|3>=|3>,\ y|3>=z|3>=-|3>,
$$
$$
x|4>=-|4>,\ y|4>=|4>,\ z|4>=-|4>.
$$
The quantum channel corresponding to the graph generated by the elements $\{1,x,y,z\}$  satisfying (\ref {kanal}) is given by the formula
\begin{equation}\label{prost}
\Phi (\rho )=(1-\alpha -\beta )\rho +\alpha x\rho x+\beta y\rho y,
\end{equation}
$\rho \in \mathfrak {S}(H_4)$, $\alpha ,\beta \ge 0,\ \alpha +\beta \le 1$. The map defined by (\ref {prost}) is a dephazing channel which is a partial case
of pseudo-diagonal channel (\ref {pseudo}) such that
$$
\Phi (|j><j|)=|j><j|,\ 1\le j\le 4,
$$
$$
\Phi (|1><2|)=(1-2\alpha -2\beta)|1><2|,\ \Phi (|1><3|)=(1-2\beta )|1><3|,
$$
$$
\Phi (|1>4|)=(1-2\alpha )|1><4|,\ \Phi (|2><3|)=(1-2\alpha )|2><3|,\
$$
$$
\Phi (|2><4|)=(1-2\beta )|2><4|,\ \Phi (|3><4|)=(1-2\alpha -2\beta )|3><4|.
$$
A transition from the graph determined by (\ref {relations})-(\ref {klein}) corresponding to the channel
(\ref {prost}) to the graph ${\mathcal L}_{\theta }$ corresponding to a general pseudo-diagonal channel (\ref {pseudo}) can be considered as a noncommutative deformation of
the Klein group $K_4$ in the spirit of \cite {Fad}.

Our goal is a explanation of subalgebras ${\cal M}_{\theta}$ of $Mat_4({\mathbb C})$ generated by subspaces ${\cal L}_{\theta}, \theta \in {\mathbb C}^*$ in terms of representation theory.  It follows from (\ref {L}) that the generators of ${\cal M}_{\theta}$ can be pick up as follows:
\begin{equation}
\label{xyz}
X =
\begin{pmatrix}
0 & 1 & 0 & 0\\
1 & 0 & 0 & 0\\
0 & 0 & 0 & 1\\
0 & 0 & 1 &0
\end{pmatrix},
Y =
\begin{pmatrix}
0 & 0 & \theta & 0\\
0 & 0 & 0 & 1/\theta\\
1/\theta & 0 & 0 & 0\\
0 & \theta & 0 & 0
\end{pmatrix},
Z =
\begin{pmatrix}
0 & 0 & 0 & 1\\
0 & 0 & 1 & 0\\
0 & 1 & 0 & 0\\
1 & 0 & 0 & 0
\end{pmatrix}
\end{equation}
Alternatively it is possible to pick up the following matrices as generators of ${\cal M}_{\theta}$ for $\theta \in {\mathbb C}^*$:
\begin{equation}
X =
\begin{pmatrix}
0 & 1 & 0 & 0\\
1 & 0 & 0 & 0\\
0 & 0 & 0 & 1\\
0 & 0 & 1 &0
\end{pmatrix},
XYX =
\begin{pmatrix}
0 & 0 & 1/\theta & 0\\
0 & 0 & 0 & \theta\\
\theta & 0 & 0 & 0\\
0 & 1/\theta & 0 & 0
\end{pmatrix},
Z =
\begin{pmatrix}
0 & 0 & 0 & 1\\
0 & 0 & 1 & 0\\
0 & 1 & 0 & 0\\
1 & 0 & 0 & 0
\end{pmatrix}
\end{equation}
It results in ${\cal M}_{\theta} = {\cal M}_{1/\theta}$. Also, we have the following symmetry ${\cal M}_{\theta} = {\cal M}_{-\theta}$.
Note that ${\rm dim}_{\mathbb C}{\cal M}_{\theta} = 8$ for $\theta \ne \pm 1$ and ${\rm dim}_{\mathbb C}{\cal M}_{\theta} = 4$ for $\theta = \pm 1$.

For explanation of this effect we will construct family of algebras ${\cal A}_{\theta}, \theta \in {\mathbb C}^*$ as a quotient of ${\mathbb C}G$ by relation:
\begin{equation}\label{sootn}
(xy + yx)z = (\theta+\theta^{-1})\cdot 1
\end{equation}
One can see that we have canonical isomorphism: ${\cal A}_{\theta} \cong {\cal A}_{1/\theta}$.
For construction of this family we will use a semidirect product of ${\mathbb Z}_2$ and ${\mathbb Z} \oplus {\mathbb Z}_2$ and its representations. Matrices $X,Y,Z$ defines representation $\phi$ of algebras ${\cal A}_{\theta}$ and ${\cal M}_{\theta} = \phi({\cal A}_{\theta})$.
Family ${\cal A}_{\theta}$ has the following symmetry ${\cal A}_{\theta} \cong {\cal A}_{-\theta}, \theta \in {\mathbb C}^*$.
Any algebra ${\cal A}_{\theta}$ has dimension 8. Our main result tells us that algebra ${\cal A}_{\theta}$ is isomorphic to $Mat_2({\mathbb C}) \oplus Mat_2({\mathbb C})$ for $\theta \ne \pm 1$. For $\theta = \pm 1$ algebras ${\cal A}_{\theta}$ have 4-dimensional radical $J$, and, quotient ${\cal A}_{\theta}/J$ is a direct sum of four copies of ${\mathbb C}$'s. Algebra ${\cal A}_{\theta}$ is "universal" for representation $\phi$, i.e.
\begin{itemize}
\item{we have factorization: for any $\theta \in {\mathbb C}^*$ representation $\phi: {\mathbb C}G \to Mat_4({\mathbb C})$ is a composition ${\mathbb C}G \to {\cal A}_{\theta} \to Mat_4({\mathbb C})$. }
\item{Also, for $\theta \ne \pm 1$ map $\phi: {\cal A}_{\theta} \to {\cal M}_{\theta}$ is an isomorphism. }
\end{itemize}

Our article is organized as follows.
In section 2 we discuss matrices $X,Y,Z$ and Klein group. In this section we remind the notion of Klein group and some property of representations of it. In section 3, we give the definition of group $G$. This group $G$ is natural object for studying of ${\cal M}_{\theta}, \theta \in {\mathbb C}^*$.
Also, we give preliminary definitions and results in group theory. In particular, we prove that group $G$ is a semidirect product of ${\mathbb Z}_2$ and ${\mathbb Z} \oplus {\mathbb Z}_2$. One can find the exploration of center of group algebra ${\mathbb C}G$ in section 4. This exploration is important for studying of representations of algebra ${\mathbb C}G$. We give a natural description of irreducible ${\mathbb C}G$-modules in terms of maximal commutative subalgebra ${\mathbb C}P \subset {\mathbb C}G$, where $P$ is a maximal abelian subgroup of $G$.
One can describe the connection between ${\cal M}_{\theta}$ and ${\mathbb C}G$ in the following manner: matrices $X,Y,Z$ defines representation of the group $G$. In section 6 we prove that ${\mathbb C}G$ - representation $\phi$ defined by matrices $X,Y,Z$ is semisimple, i.e. direct sum of irreducible ${\mathbb C}G$ - modules.
Also, we introduce the family of algebras ${\cal A}_{\theta}$. This family is a quotient of ${\mathbb C}G$ by relations (\ref{sootn}). Family ${\cal A}_{\theta}$ plays an important role for studying of representation $\phi$. Actually, morphism $\phi: {\mathbb C}G \to Mat_4({\mathbb C})$ factorizes into composition: ${\mathbb C}G \to {\cal A}_{\theta} \to Mat_4({\mathbb C})$. Thus, $\phi({\cal A}_{\theta}) = {\cal M}_{\theta}$.
We prove that if $\theta \ne \pm 1$ algebra ${\cal A}_{\theta}$ is isomorphic to direct sum of two copies of matrix algebras $Mat_2({\mathbb C})$. In the case $\theta = \pm 1$, algebra ${\cal A}_{\theta}$ has 4-dimensional radical. This means that relation (\ref{sootn}) is universal, i.e. for general $\theta \in {\mathbb C}^*$ representation $\phi$ provides isomorphism. In the case $\theta = \pm 1$, representation $\phi$ is not an isomorphism, because algebras ${\cal A}_{\theta}$ has radical. Since $\phi$ is semisimple, $\phi$ annihilate radical and we get 4-dimensional algebra ${\cal M}_{\theta}$.
Article has two appendices. In the beginning of appendix A, we give some classical definitions in homological algebra. In this section we study homological properties of irreducible ${\mathbb C}G$-modules. In particular, we calculate extension group for irreducible ${\mathbb C}G$-modules.
In Appendix B, we recall some definitions and notions of noncommutative algebraic geometry. We introduce two kinds of noncommutative varieties for arbitrary finite generated algebra $A$: representation space ${\bf Rep}_nA = {\rm Hom}_{alg}(A,Mat_n({\mathbb C}))$ and moduli variety ${\cal M}_nA$ which is the quotient of ${\bf Rep}_nA$ by natural action of ${\rm GL}_n({\mathbb C})$. In this section we study connection between varieties ${\bf Rep}_2{\mathbb C}P$ and ${\bf Rep}_2{\mathbb C}G$ as well as ${\cal M}_2{\mathbb C}P$ and ${\cal M}_2{\mathbb C}G$.

\section{Matrices $X,Y,Z$ and Klein's group.}
\label{kleingr}

In this section we will introduce matrices $X,Y,Z$ depending on $\theta$ and consider partial cases $\theta = \pm 1$. In this case we will show that these matrices define representation of Klein's group.

Consider matrices $X$, $Y$ and $Z$ from the formula (\ref{xyz}).
Denote by $I$ the identity matrix.
It follows that ${\cal L}_{\theta} = a I + b X + c Y + d Z, a,b,c,d \in {\mathbb C}$ for any $\theta \in {\mathbb C}^*$. Moreover, $X,Y,Z$ satisfy (\ref {relations}).
These relations motivate us to interpret subspace ${\cal L}_{\theta}$ as an image of admissible $4$-dimensional representations of the group $G$. 

One can see that if $\theta = 1$, then $X,Y,Z$ have the following view:
\begin{equation}
\begin{pmatrix}
0 & 1 & 0 & 0\\
1 & 0 & 0 & 0\\
0 & 0 & 0 & 1\\
0 & 0 & 1 &0
\end{pmatrix},
\begin{pmatrix}
0 & 0 & 1 & 0\\
0 & 0 & 0 & 1\\
1 & 0 & 0 & 0\\
0 & 1 & 0 & 0
\end{pmatrix},
\begin{pmatrix}
0 & 0 & 0 & 1\\
0 & 0 & 1 & 0\\
0 & 1 & 0 & 0\\
1 & 0 & 0 & 0
\end{pmatrix}
\end{equation}
respectively.
In this case our matrices are commuting. And hence, we have the following relations:
\begin{equation}
\label{commuti}
XY = YX, XZ = ZX, YZ = ZY, X^2 = Y^2 = Z^2 = 1.
\end{equation}
Also, one can check that
\begin{equation}
\label{xyz}
Z = XY.
\end{equation}

If we consider group with generators $x,y$ satisfying to relations (\ref{commuti}) and (\ref{xyz}), we get the well-known Klein's group $K_4$ of order 4 isomorphic to ${\mathbb Z}_2 \oplus {\mathbb Z}_2$, where first ${\mathbb Z}_2$ is generated by $x$, second one is generated by $y$.
Thus, in the case of $\theta = 1$ matrices $X,Y$ define representation $\phi$ of $K_4$ by rule: $\phi: x \mapsto X, y \mapsto Y$.
Recall that representation $\phi$ is a direct sum of irreducible ones. As for any abelian group, all irreducible representation of $K_4$ is 1-dimensional.
These one-dimensional representations (or characters) are elements of group ${\rm Hom}(K_4, {\mathbb C}^*)$, where ${\mathbb C}^*$ is an algebraic torus by the definition.
By Pontryagin's duality, group ${\rm Hom}(K_4, {\mathbb C}^*) \cong K_4$. Thus, we can describe characters in the following manner:
$$
\chi_1: x \mapsto 1, y \mapsto 1
$$
$$
\chi_x: x \mapsto -1, y \mapsto 1
$$
$$
\chi_y: x \mapsto 1, y \mapsto -1
$$
$$
\chi_{xy}: x \mapsto -1, y \mapsto -1
$$
One can check that representation $\phi$ is isomorphic to direct sum $\chi_1 \oplus \chi_x \oplus \chi_y \oplus \chi_{xy}$.
It means that there is a basis in which matrix $X$ and $Y$ are diagonal and matrices $X$ and $Y$ have eigenvalues $\chi_1(x),\chi_x(x),\chi_y(x),\chi_{xy}(x)$ and $\chi_1(y),\chi_x(y),\chi_y(y),\chi_{xy}(y)$ respectively.

Also, consider the case of $\theta = -1$. 
One can check that if $\theta = -1$ matrices $X,Y,Z$ are commuting and $Z = -XY$. Of course, matrices $X, -Y$ as images of generators $x,y$ of $K_4$ define representation $\rho: x \mapsto X, y \mapsto -Y$ which is similar to the case $\theta =1$.

In the next sections we will generalize description of ${\cal L}_{\pm 1}$ to the case of arbitrary $\theta$ by means of representation theory methods.

\section{Some facts from group theory and description of the group $G$.}

In this section we will study group $G$ with generators $x,y,z$ and relations (\ref{relations}).

Firstly, recall the notion of a free group $F(S)$ with a set of
generators
$S$. Suppose that $S = \{s_1,...,s_k\}$.
Denote $J$ the set $S \cup S^{-1}$. Define {\it word} as a product of
elements of $J$. Word can be simplified by
omitting consequent symbols $s$ and $s^{-1}$. Word which cannot be
simplified is called {\it reduced}.
The set of reduced words equipped by the operation of concatenation of
words
is said to be a free group $F(S)$ of {\it rank $k$}.  Finite generated group $G$ is a quotient of free group with arbitrary set of generators. We will say that group $H$ has {\it presentation} $H = \langle S | R \rangle$ where $S$ is a set of generators, $R$ is a set of relations. One can use the following description of $R$. Any relation we will rewrite in the view: $r = 1$, where $r \in F(S)$.
Group $H$ is a quotient $F(S)/N(R)$, where $N(R)$ is a {\it normal closure} in free group $F(S)$ of the set of relations $R$, i.e. minimal normal subgroup of $F(S)$ containing $R$. Of course, presentation of fixed group $H$ is not unique.

Recall the construction of free product of groups. Assume that we have two groups $H_1$, $H_2$ which have the following presentations: $H_1 = \langle S_1 | R_1 \rangle$ and $H_2 = \langle S_2 | R_2 \rangle$ with different sets of generators $S_1$ and $S_2$. Define free product $H_1 * H_2$ as group with presentation $\langle S_1 \cup S_2 | R_1 \cup R_2 \rangle$. One can show that notion of free product is well-defined. Analogously, one can define free product $H_1 * ... * H_k$ of the groups $H_1,...,H_k$.

Assume that $S = \{x,y,z\}$. One can consider group $G$ as quotient of $F(S)$. There are relations $x^2 = y^2 = z^2 = 1, xz=zx, yz=zy$. Rewrite these relations in the following manner: $x^2 = y^2 = z^2 = 1, xzx^{-1}z^{-1} = yzy^{-1}z^{-1} = 1$. Thus, $R = \{ x^2, y^2, z^2, xzx^{-1}z^{-1}, yzy^{-1}z^{-1} \}$. 
Consider set of generators $S = \{x,y,z\}$, set of relations: $R_1 = \{ x^2,y^2,z^2 \} \subset R$, normal subgroup $N(R_1)$ and quotient $P = F(S)/N(R_1)$. The group $P$ has the following presentation: $P = \langle S | R_1 \rangle$. It is clear that
\begin{equation}
P = \langle x,y,z | x^2,y^2,z^2 \rangle = \langle x| x^2\rangle * \langle y| y^2\rangle * \langle z| z^2\rangle \cong {\mathbb Z}_2 * {\mathbb Z}_2 * {\mathbb Z}_2
\end{equation}
where ${\mathbb Z}_2$ is a cyclic group of order 2. In this way we get that group $G$ is a quotient of $P$ by the relations $xzx^{-1}z^{-1} = 1, yzy^{-1}z^{-1} = 1$.

Consider natural morphism: $\phi: {\mathbb Z}_2 * {\mathbb Z}_2 * {\mathbb Z}_2 \to {\mathbb Z}_2 \oplus {\mathbb Z}_2 \oplus {\mathbb Z}_2$.
Using results on subgroup of free product (independently Kurosh \cite{kurosh}, Baer and Levi \cite{baerlevi}, Takahasi \cite{takahasi}), we get that kernel of $\phi$ is a free group $F_5$ of rank 5. It will be convenient to choose generators of the ${\rm Ker}\phi = F_5$ as follows: $xyx^{-1}y^{-1}, xzx^{-1}z^{-1}, yzy^{-1}z^{-1}, xyzy^{-1}z^{-1}x^{-1}, yxzx^{-1}z^{-1}y^{-1}$. Of course, using relations $x^2 = y^2 = z^2 = 1$, we obtain that $xyx^{-1}y^{-1} = xyxy, xzx^{-1}z^{-1} = xzxz, yzy^{-1}z^{-1} = yzyz, xyzy^{-1}z^{-1}x^{-1} = xyzyzx, yxzx^{-1}z^{-1}y^{-1} = yxzxzy$.
One can show that group $G$ is a quotient of $P$ by the relations $R_2 = \{xzx^{-1}z^{-1} = xzxz = 1, yzy^{-1}z^{-1} = yzyz = 1\}$. Consider normal closure $N(R_2)$ in the group $P$ of the set of the elements $R_2$. 
Consider quotient of ${\mathbb Z}_2 * {\mathbb Z}_2 * {\mathbb Z}_2$ by normal subgroup $N(R_2)$.

Let us prove the following fact:
\begin{Proposition}
We have the following exact sequence for non-abelian group $G$:
\begin{equation}
\label{longdiag}
\xymatrix{
0\ar[r] & {\mathbb Z}\ar[r] & G\ar[r] & {\mathbb Z}_2 \oplus {\mathbb Z}_2 \oplus {\mathbb Z}_2 \ar[r] & 1,
}
\end{equation}
where normal subgroup ${\mathbb Z}$ of the group $G$ is generated by element $xyxy$.
\end{Proposition}
\begin{proof}
We have the following exact sequence for free product ${\mathbb Z}_2 * {\mathbb Z}_2 * {\mathbb Z}_2$:
\begin{equation}
\xymatrix{
1 \ar[r] & F_5 \ar[r] & {\mathbb Z}_2 * {\mathbb Z}_2 * {\mathbb Z}_2 \ar[r] & {\mathbb Z}_2 \oplus {\mathbb Z}_2 \oplus {\mathbb Z}_2 \ar[r] & 1,
}
\end{equation}
As we know, free group $F_5$ is generated by $xyxy, xzxz, yzyz, xyzyzx, yxzxzy$. Group $N(R_2)$ is a normal closure of $xzxz, yzyz$. If $xzxz = 1$ and $yzyz = 1$ then $xyzyzx = 1$ and $yxzxzy = 1$. Denote by $a_1,...,a_5$ the generators $xyxy, xzxz, yzyz, xyzyzx, yxzxzy$ of $F_5$. Denote by $F'_5 = [F_5,F_5]$ commutant of $F_5$. Consider natural morphism $f: F_5 \to F_5/F'_5 = {\mathbb Z}^{\oplus 5}$. Denote by $a_i, i = 1,...,5$ the generators of ${\mathbb Z}^{\oplus 5}$ which are images of $xyxy, xzxz, yzyz, xyzyzx, yxzxzy$ under $f$ respectively. Also, we have projection ${\mathbb Z}^{\oplus 5} \to {\mathbb Z}$ defined by rule: $\sum k_i a_i \mapsto k_1 a_1$. Consider composition $F_5 \to {\mathbb Z}$. One can check that $N(R_2)$ is a kernel of this composition.

Moreover, we get the following commutative diagram:
\begin{equation}
\label{comdiagf}
\xymatrix{
0 \ar[r] & {\mathbb Z}\ar[r] & G\ar[r] & {\mathbb Z}_2 \oplus {\mathbb Z}_2 \oplus {\mathbb Z}_2 \ar[r] & 1\\
1 \ar[r] & F_5 \ar[r]\ar[u] & {\mathbb Z}_2 * {\mathbb Z}_2 * {\mathbb Z}_2 \ar[r]\ar[u] & {\mathbb Z}_2 \oplus {\mathbb Z}_2 \oplus {\mathbb Z}_2 \ar[r]\ar[u]^{\cong} & 1\\
1 \ar[r] & N(R_2)\ar[r]\ar[u] & N(R_2)\ar[r]\ar[u] & 1\ar[u]
}
\end{equation}
\end{proof}

Exact sequence (\ref{longdiag}) tells us that group $G$ is the extension of ${\mathbb Z}_2 \oplus {\mathbb Z}_2 \oplus {\mathbb Z}_2$ by ${\mathbb Z}$. Automorphism group ${\rm Aut}({\mathbb Z}) \cong {\mathbb Z}_2$, i.e. there are only two automorphisms of ${\mathbb Z}$: trivial and morphism defined by correspondence: $n \mapsto -n, n \in {\mathbb Z}$. It is clear that we have a well-defined action of ${\mathbb Z}_2 \oplus {\mathbb Z}_2 \oplus {\mathbb Z}_2$ on ${\mathbb Z}$ by conjugation. This action defines homomorphism of groups: ${\mathbb Z}_2 \oplus {\mathbb Z}_2 \oplus {\mathbb Z}_2 \to {\rm Aut}({\mathbb Z}) = {\mathbb Z}_2$. One can calculate that kernel of this morphism is a subgroup ${\mathbb Z}_2 \oplus {\mathbb Z}_2$ generated by images of $xy$ and $z$ under natural morphism $G \to {\mathbb Z}_2 \oplus {\mathbb Z}_2 \oplus {\mathbb Z}_2$. Denote by $P$ subgroup of $G$ generated by $xy$ and $z$.
We have the following commutative diagram:
\begin{equation}
\xymatrix{
&& {\mathbb Z}_2 \ar[r]^{=} & {\mathbb Z}_2\ar[r] & 1\\
0 \ar[r] & {\mathbb Z}\ar[r] & G \ar[r]\ar[u] & {\mathbb Z}_2 \oplus {\mathbb Z}_2 \oplus {\mathbb Z}_2 \ar[r]\ar[u] & 1\\
0 \ar[r] & {\mathbb Z}\ar[r]\ar[u]^{=} & P \ar[r]\ar[u] & {\mathbb Z}_2 \oplus {\mathbb Z}_2 \ar[r]\ar[u] & 1
}
\end{equation}

It follows that group $P$ is abelian. Also, one can show that $P \cong {\mathbb Z} \oplus {\mathbb Z}_2$.

Recall the notion of {\it semidirect} product. Consider group $H$. Assume that there are normal subgroup $N_1$ and subgroup $H_1$ of $H$. Also, assume that composition of embedding $H_1 \hookrightarrow H$ and projection $H \to H/N_1$ is an isomorphism $H_1 \cong H/N_1$. In this case, we will say that $H$ is a semidirect product of $N_1$ and $H_1$. We will denote by $N_1 \rtimes H_1$ the semidirect product of $N_1$ and $H_1$. Equivalently, assume that there is an exact sequence:
\begin{equation}
\label{semi}
\xymatrix{
1 \ar[r] & N_1 \ar[r] & H \ar[r] & H/N_1 \ar[r] & 1.
}
\end{equation}
$H$ is a semidirect product of $N_1$ and $H_1$ iff sequence (\ref{semi}) is split, i.e. there is a section $H/N_1 \to H$ of the natural morphism
$H \to H/N_1$. Also, one can say that there is a subgroup $H_1$ such that $H_1 \cong H/N_1$ and $H_1 \cap N_1 = \{1\}$.

\begin{Proposition}
Group $G$ is a semidirect product $({\mathbb Z} \oplus {\mathbb Z}_2) \rtimes {\mathbb Z}_2$.
\end{Proposition}
\begin{proof}
We have the following exact sequence:
\begin{equation}
\xymatrix{
0 \ar[r] & P = {\mathbb Z} \oplus {\mathbb Z}_2\ar[r] & G\ar[r] & {\mathbb Z}_2 \ar[r] & 1
}
\end{equation}
Recall that subgroup $P = {\mathbb Z} \oplus {\mathbb Z}_2$ is generated by $xy$ and $z$.
This sequence has section ${\mathbb Z}_2 \to G$ defined by element $x$. Therefore, $G$ is a semidirect product $({\mathbb Z} \oplus {\mathbb Z}_2) \rtimes {\mathbb Z}_2$.
\end{proof}

\section{Center of group algebra ${\mathbb C}G$.}
\label{cent}

Let us pick up generators of group $G = ({\mathbb Z} \oplus {\mathbb Z}_2) \rtimes {\mathbb Z}_2$ as follows: $g = xy$, $x$ and $z$. Elements $g,x,z$ are generators because we have the identity: $y = xg$. Element $g$ has infinite order. In this case group $P = {\mathbb Z} \oplus {\mathbb Z}_2$ is generated by $g$ and $z$. We will use multiplicative description of $\mathbb Z$. Thus, we have the following relations: $x^2 = z^2 = 1, xz = zx, zg = gz, xgx = g^{-1}$.
In this section we will study representation theory of group $G$.

Since group $P$ is a subgroup of index 2 in the group $G$, we can formulate the following useful fact:
\begin{Proposition}
Any element $w \in G$ can be uniquely written as follows: $w = g'$ or $w = g'x$ for arbitrary element $g' \in P$.
\end{Proposition}

Recall the notion of free $A$-module for arbitrary algebra $A$. $A$-module $M$ is a free $A$-module of rank $n$ iff $M$ is isomorphic to direct sum of $n$ copies of $A$: $M \cong \oplus^n_{i=1}A$.
\begin{Corollary}
Group algebra ${\mathbb C}G$ is a free ${\mathbb C}P$ - module of rank 2.
\end{Corollary}
\begin{proof}
Using decomposition $G = P \cup Px$, we get the decomposition of ${\mathbb C}P$-module ${\mathbb C}G$ into direct sum ${\mathbb C}P \oplus {\mathbb C}P$. Also, we can choose element $1,x \in G$ as generators of these free ${\mathbb C}P$-modules.
\end{proof}

Group algebra ${\mathbb C}P$ is an extended algebra of Laurent's polynomials ${\mathbb C}[g,g^{-1},z; z^2 = 1]$.
Consider center $\cal C$ of group algebra ${\mathbb C}G$. Center ${\cal C}$ is a subalgebra of ${\mathbb C}P$ containing elements commuting with $x$.
Equivalently, element $c \in {\cal C}$ iff $xcx^{-1} = xcx = c$.
As we know, we have well-defined action of $G$ on $P$ by conjugation. Therefore, we have involution $c_x$ defined by rule:
\begin{equation}
\label{invcx}
c_x: (g,z) \mapsto (xgx^{-1},xzx^{-1}) = (xgx, xzx) = (g^{-1},z).
\end{equation}
Using this involution, we get that ${\cal C}$ is a subalgebra of ${\mathbb C}G$ consisting of $c_x$ - invariant elements.

It can be shown that $\cal C$ is a subalgebra of ${\mathbb C}P$ generated by $g + g^{-1}$ and $z$. Thus,
\begin{equation}\label{dop}
{\cal C} \cong {\mathbb C}[u = g+g^{-1}] \otimes {\mathbb C}[z, z^2 = 1].
\end{equation}

Consider ${\mathbb C}P$ as ${\cal C}$ - module. Prove the following useful fact.
\begin{Proposition}
\label{freemod}
${\mathbb C}P$ is a free ${\cal C}$ - module of rank 2.
\end{Proposition}
\begin{proof}
It is sufficient to prove that ${\mathbb C}[g,g^{-1}]$ is a free ${\cal C}$-module of rank 2.
For this purpose, consider fixed polynomial $f(g) \in {\mathbb C}[g,g^{-1}]$. Prove that there are uniquely determined elements $z_1, z_2 \in {\cal C}$ such that
\begin{equation}
f(g) = z_1 + gz_2.
\end{equation}
Multiply this identity by $x \in G$ from both sides.
We get the following identity:
\begin{equation}
xf(g)x = f(g^{-1}) = z_1 + g^{-1}z_2
\end{equation}
Thus, we get that
\begin{equation}
f(g) - f(g^{-1}) = (g-g^{-1})z_2.
\end{equation}
It follows that $z_2 = \frac{f(g) - f(g^{-1})}{g - g^{-1}} \in {\cal C}$ is a Laurent's polynomial over $g$.
Also, $z_1 = f(g) - gz_2$. Therefore, elements $z_1$ and $z_2$ are uniquely determined by $f(g)$ and, hence, ${\mathbb C}P$ is a free $\cal C$ - module of rank 2.
\end{proof}

\begin{Corollary}
\label{dimalg}
Algebra ${\mathbb C}G$ is a free ${\cal C}$ - module of rank 4.
\end{Corollary}

Remind the following useful notion of tensor product of modules over algebra. Let $V_1$ and $V_2$ be right and left $A$-module respectively for some algebra $A$. Denote by $V_1 \otimes_A V_2$ the quotient of $V_1 \otimes V_2$ by relations $v_1 a \otimes v_2 - v_1 \otimes a v_2$ for any $v_1 \in V_1, v_2 \in V_2, a \in A$. If $V_1$ has structure of left $A$ - module, then $V_1 \otimes_A V_2$ is $A$-module.

For fixed algebra and its representation, the space of representation is a module. If representation is irreducible, then we will say that corresponding module is irreducible. Further, let us study irreducible ${\mathbb C}G$ - modules.

Since $\cal C$ is a commutative algebra, we can consider ${\mathbb C}G$ as algebra over ${\cal C}$. Fix character $\chi \in {\rm Spec}{\cal C}$. Denote ${\cal C}$-module corresponding to $\chi$ by ${\mathbb C}^{\chi}$. It can be shown that ${\mathbb C}G \otimes_{\cal C} {\mathbb C}^{\chi}$ is an algebra over $\mathbb C$. One can check that algebra ${\mathbb C}G \otimes_{\cal C} {\mathbb C}^{\chi}$ is a quotient of ${\mathbb C}G$ by relations $z - \chi(z)\cdot 1$.

By Schur's lemma, for any central element $c \in {\cal C}$ and any irreducible representation $\rho$ of algebra ${\mathbb C}G$ matrix
$\rho(c)$ is scalar, i.e. $\rho(c) = \chi(c)1$ for some character $\chi$ of ${\cal C}$.
Thus, we get that any irreducible ${\mathbb C}G$ - module correspond to some character of ${\cal C}$. Set of characters of ${\cal C}$ is a variety ${\rm Spec} {\cal C} = {\rm Hom}_{alg}({\cal C}, {\mathbb C})$. This variety is called {\it variety of characters} of the algebra ${\cal C}$.

Let us give some description of characters of ${\cal C}$.
Any $f \in {\rm Hom}_{alg}({\cal C}, {\mathbb C})$ is defined by pair of values $f(u) = f(g+g^{-1})$ and $f(z)$. One can show that $f(u) \in {\mathbb C}$ and because of $z^2 = 1$ we get that $f(z) = \pm 1$.  
Thus, ${\rm Spec} {\cal C}$ has two components. Each of them is an affine line ${\mathbb C}^1$. Denote these components by ${\mathbb C}^1_{-}$ and ${\mathbb C}^1_{+}$ corresponding to values $f(z) = -1$ and $f(z) = 1$ respectively.

All irreducible ${\mathbb C}G$-modules correspond to fixed character $\chi$ are representations of ${\mathbb C}G \otimes_{\cal C} {\mathbb C}^{\chi}$. Thus, we have to study algebra ${\mathbb C}G \otimes_{\cal C} {\mathbb C}^{\chi}$ for description of irreducible ${\mathbb C}G$ - modules.
For studying of algebra ${\mathbb C}G \otimes_{\cal C} {\mathbb C}^{\chi}$, let us remind that Jacobson radical (or radical) of the algebra $A$ is an intersection of maximal right ideals of $A$. Equivalent formulation: Jacobson radical consists of all elements annihilated by all simple left $A$-modules. Also, note that for finite-dimensional algebras we have the following description: Jacobson radical is a maximal nilpotent ideal.

\begin{Proposition}
\label{algirr}
If $\chi(u) = \chi(g+g^{-1}) \ne \pm 2$, then algebra ${\mathbb C}G \otimes_{\cal C} {\mathbb C}^{\chi}$ is isomorphic to ${\rm Mat}_2({\mathbb C})$. If $\chi(u) = \chi(g+g^{-1}) = \pm 2$, then algebra ${\mathbb C}G \otimes_{\cal C} {\mathbb C}^{\chi}$ has 2-dimensional radical.
\end{Proposition}
\begin{proof}
Using corollary \ref{dimalg}, we get that ${\rm dim}_{\mathbb C}{\mathbb C}G \otimes_{\cal C} {\mathbb C}^{\chi} = 4$. It can be shown that this algebra has the following basis: $1,x,g,xg$. There are relations: $x^2 = 1, g+g^{-1} = \chi(u)1, xgx = g^{-1}, \chi(u) \in {\mathbb C}$. Denote by $t$ the root of equation: $t+t^{-1} = \chi(u)$. We have the following morphism: $F: {\mathbb C}G \otimes_{\cal C} {\mathbb C}^{\chi} \to {\rm Mat}_2({\mathbb C})$ defined by rule:
\begin{equation}
F: x \mapsto
\begin{pmatrix}
0 & 1\\
1 & 0
\end{pmatrix}
,
g \mapsto
\begin{pmatrix}
t & 0\\
0 & t^{-1}
\end{pmatrix}
\end{equation}
One can check that if $\chi(u) \ne \pm 2$, then $F$ is an isomorphism. If $\chi(u) = 2$, then $g+g^{-1} = 2$ and, hence, $(g-1)^2 = 0$. Consider two-sided ideal $J$ of ${\mathbb C}G \otimes_{\cal C} {\mathbb C}^{\chi}$ generated by $g-1$. It can be shown that $J$ has a basis $x(g-1), g-1$ over ${\mathbb C}$. Also, consider $J^2 = \langle j_1 \cdot j_2| j_1, j_2 \in J \rangle$. One can check that $J^2 = 0$. Thus, ${\rm dim}_{\mathbb C}J = 2$. Note that element $x$ and $1$ are not in Jacobson radical, because these element are not nilpotent. Thus, $J$ is a maximal nilpotent ideal of ${\mathbb C}G \otimes_{\cal C} {\mathbb C}^{\chi}$. Therefore, $J$ is a Jacobson radical.
\end{proof}

\begin{Corollary}
\label{irr}
Dimension of irreducible representation of group $G$ is less or equal to 2.
\end{Corollary}
\begin{proof}
Irreducible representations of ${\mathbb C}G$ are irreducible representations of ${\mathbb C}G \otimes_{\cal C} {\mathbb C}^{\chi}$ for some $\chi \in {\rm Spec}{\cal C}$.  Using proposition \ref{algirr}, we get that irreducible representations of ${\mathbb C}G \otimes_{\cal C} {\mathbb C}^{\chi}$ is two-dimensional if $\chi(u) \ne \pm 2$. If $\chi(u) = \pm 2$, then irreducible representations of ${\mathbb C}G$ are one-dimensional.
\end{proof}

\section{Description and properties of irreducible ${\mathbb C}G$ - modules.}

In this section we will describe irreducible ${\mathbb C}G$ - modules in terms of ${\mathbb C}P$ - modules.

Assume that $V$ is an irreducible ${\mathbb C}G$ - module. It follows from Corrolary \ref{irr} that ${\rm dim}_{\mathbb C}V \le 2$. Also, recall that algebra
${\mathbb C}P$ is an extended algebra of Laurent polynomials: ${\mathbb C}[g^{\pm 1}] \otimes {\mathbb C}[z, z^2 = 1]$.

Consider variety ${\rm Spec} {\mathbb C}P = {\rm Hom}_{alg}({\mathbb C}[g^{\pm 1}] \otimes {\mathbb C}[z, z^2 = 1],{\mathbb C})$. Then, $\chi(z) = \pm 1$ for any $\chi \in {\rm Spec} {\mathbb C}P$.
Thus, variety ${\rm Spec} {\mathbb C}P$ has two components, each of them is an one-dimensional algebraic torus ${\mathbb C}^*$. Components correspond to $\chi(z) = \pm 1$. We will denote these components by ${\mathbb C}^*_{-}$ and ${\mathbb C}^*_{+}$.
We have the following description of $\chi \in {\rm Spec} {\mathbb C}P$: $\chi(g,z) = (\chi(g) = a, \chi(z) = \pm 1) \in {\mathbb C}^* \times \{{\pm 1}\}$. Remind that ${\cal C}$ is an algebra of invariants of ${\mathbb C}P$ under action of involution $c_x$ defined by rule (\ref{invcx}). Also, with action of $c_x$ on ${\mathbb C}P$, we have induced action of $c_x$ on ${\rm Spec}{\mathbb C}P$. For character $\chi \in {\rm Spec}{\mathbb C}P$, we have the following identity: $c_x(\chi)(g) = \chi(xgx) = \chi(g^{-1}) = \chi^{-1}(g)$, i.e. $c_x(\chi) = \chi^{-1}$.
We have the following surjective morphism of varieties:
\begin{equation}
pr:{\rm Spec}{\mathbb C}P \to {\rm Spec}{\cal C}
\end{equation}
defined by rule:
$pr:(a, 1) \mapsto (a+a^{-1}, 1)$ and $pr:(a, -1) \mapsto (a+a^{-1}, -1)$. One can check that involution $c_x$ acts on the fibers of the morphism $pr$.
Also, one can say that variety ${\rm Spec}{\cal C}$ is a quotient of ${\rm Spec}{\mathbb C}P$ by involution $c_x$. Recall that ${\rm Spec}{\mathbb C}P$ and ${\rm Spec}{\cal C}$ are disjoint unions ${\mathbb C}^*_{-} \cup {\mathbb C}^*_{+}$ and ${\mathbb C}^1_{-} \cup {\mathbb C}^1_{+}$ respectively. Components ${\mathbb C}^*_{\pm}$ and ${\mathbb C}^1_{\pm}$ correspond to different values of $\chi(z) = \pm 1$. Morphism $pr$ decomposes into $pr: {\mathbb C}^*_{-} \to {\mathbb C}^1_{-}$ and $pr: {\mathbb C}^*_{+} \to {\mathbb C}^1_{+}$.

Consider ideal $I$ of ${\mathbb C}P$ generated by elements $(g - \chi(g) \cdot 1)^k$ and $z - \chi(z) \cdot 1$.
Denote by ${\mathbb C}^{\chi}(k)$ the quotient of ${\mathbb C}P$-module ${\mathbb C}P/I$. It follows that ${\mathbb C}^{\chi} = {\mathbb C}^{\chi}(1)$. 
Consider finite-dimensional ${\mathbb C}G$ - module $V$. Denote by $V_{\chi}(k)$ the following subspace of $V$:
\begin{equation}
V_{\chi}(k) = \{v \in V| (g - \chi(g)\cdot 1)^k v = 0, (z - \chi(z)\cdot 1) v = 0\}
\end{equation}
It follows that $V_{\chi}(k) \subseteq V_{\chi}(k+1)$ for any integer k. Also, note that if $V_{\chi}(1) = 0$ then $V_{\chi}(k) = 0$ for any $k$.
Actually, if we consider restriction of $g$ on $V_{\chi}(k)$, then this restriction has eigenvector. Thus, if $V_{\chi}(k) \ne 0$ then $V_{\chi}(1) \ne 0$.

Since $V$ is finite-dimensional, we get that there is a minimal $k_0$ such that $V_{\chi}(k_0) = V_{\chi}(m)$ for $m \ge k_0$. Denote by ${\rm Char}(V)$ the set of characters $\chi \in {\rm Spec}{\mathbb C}P$ such that $V_{\chi}(1) \ne 0$.

We have the following famous fact:
\begin{Proposition}
\begin{itemize}
\item{
Consider ${\mathbb C}P$-module $V$ such that ${\rm Char}(V) = \{\chi\}$. Then there is a decomposition 
\begin{equation}
V = {\mathbb C}^{\chi}(k_1) \oplus ... \oplus {\mathbb C}^{\chi}(k_{q}).
\end{equation}
for arbitrary integers $k_1,...,k_q$.
}
\item{
Consider finite-dimensional ${\mathbb C}P$ - module $W$. There are ${\mathbb C}P$-modules $V_1,...,V_s$ and characters $\chi_1,...,\chi_s$ such that
\begin{equation}
W \cong V_1 \oplus ... \oplus V_s
\end{equation}
and ${\rm Char}V_i = \{\chi_i\}, i = 1,...,s$.
}
\end{itemize}
\end{Proposition}

\begin{Proposition}
\label{symmiv}
Consider finite-dimensional ${\mathbb C}G$ - module $V$. We have the following identity:
$$
{\rm dim}_{\mathbb C}V_{\chi}(k) = {\rm dim}_{\mathbb C}V_{c_x(\chi)}(k)
$$
for any $k$. In particular, $\chi \in {\rm Char}(V)$ iff $x(\chi) = \chi^{-1} \in {\rm Char}(V)$ and ${\rm Char}({\mathbb C}G \otimes_{{\mathbb C}P} {\mathbb C}^{\chi}) = (\chi, \chi^{-1})$
\end{Proposition}
\begin{proof}
Assume that $v \in V_{\chi}(k)$. Consider vector $xv$.  We have the following identities:
\begin{equation}
x(g - \chi(g)\cdot 1)^k v = (g^{-1} - \chi(g) \cdot 1)^k xv = (-1)^kg^{-k}(g - \chi^{-1}(g)\cdot 1)^k xv = 0
\end{equation}
Thus, $x (V_{\chi}(k)) = V_{c_x(\chi)}(k)$.
One can check that ${\rm Char}({\mathbb C}G \otimes_{{\mathbb C}P} {\mathbb C}^{\chi}) = (\chi, \chi^{-1})$. In fact, assume that $v \in {\mathbb C}^{\chi}$ such that $g v = \chi(g) v$. Thus, we can choose the following basis of ${\mathbb C}G \otimes_{{\mathbb C}P} {\mathbb C}^{\chi}$: $1 \otimes v$ and $x \otimes v$. And one can check that $V_{\chi}$ and $V_{c_x(\chi)}$ are one-dimensional spaces ${\mathbb C}(1 \otimes v)$ and ${\mathbb C}(x \otimes v)$ respectively.
\end{proof}
\begin{Corollary}
\label{isomor}
We have the following isomorphism: ${\mathbb C}G \otimes_{{\mathbb C}P} {\mathbb C}^{x(\chi)} \cong {\mathbb C}G \otimes_{{\mathbb C}P} {\mathbb C}^{\chi}$. Also, ${\mathbb C}G \otimes_{{\mathbb C}P} {\mathbb C}^{\chi}$ as ${\mathbb C}P$-module is isomorphic to direct sum ${\mathbb C}^{\chi} \oplus {\mathbb C}^{c_x(\chi)}$.
\end{Corollary}

Thus, we can deduce the following
\begin{Corollary}
Consider irreducible ${\mathbb C}G$ - module $V$. Assume that $\chi \in {\rm Char}(V)$ and $\chi(g) \ne \pm 1$. Then ${\mathbb C}G$ - module $V$ is isomorphic to ${\mathbb C}G \otimes_{{\mathbb C}P} {\mathbb C}^{\chi}$.
\end{Corollary}
\begin{proof}
Firstly, consider $V = {\mathbb C}G \otimes_{{\mathbb C}P} {\mathbb C}^{\chi}$.
Assume that $W \subseteq V$ is a ${\mathbb C}G$ - submodule.
Consider ${\rm Char}(W)$. Clearly, ${\rm Char}(W) \subseteq {\rm Char}(V)$. Since $\chi(g) \ne \pm 1$, then $c_x(\chi)(g) \ne \chi(g)$, i.e. $c_x(\chi) = \chi^{-1} \ne \chi$.
Using proposition \ref{symmiv}, we get that ${\rm Char}(W) = {\rm Char}(V)$ and, hence, ${\rm dim}_{\mathbb C}W = 2$. Thus, $W = V$.

Conversely, let $V$ be irreducible ${\mathbb C}G$ - module. As we know from corollary \ref{irr}, ${\rm dim}_{\mathbb C}V \le 2$. If $\chi(g) \ne \pm 1$, then $\chi^{-1} \ne \chi$ and ${\rm Char}(V) = (\chi, \chi^{-1})$. Thus, ${\rm dim}_{\mathbb C}V = 2$. Let $v$ be a vector such that $g v = \chi(g) v$. One can show that $v$ and $xv$ generate ${\mathbb C}G$ - submodule isomorphic to ${\mathbb C}G \otimes_{{\mathbb C}P} {\mathbb C}^{\chi}$.
Using irreducibility of $V$, we get the required.
\end{proof}

Consider the case $\chi(g) = \pm 1$. We can prove the following corollary:
\begin{Corollary}
\label{cpmf}
Assume that $\chi(g) = \pm 1$. In this case ${\mathbb C}G$ - module $V = {\mathbb C}G \otimes_{{\mathbb C}P}{\mathbb C}^{\chi}$ is a direct sum of one dimensional
modules: ${\mathbb C}^{\chi}_{1} \oplus {\mathbb C}^{\chi}_{-1}$, where ${\mathbb C}^{\chi}_{1}$ and ${\mathbb C}^{\chi}_{-1}$ are subspaces generated by eigenvector of $x$ corresponding to eigenvalues $1$ and $-1$ respectively.
\end{Corollary}
\begin{proof}
In the case $\chi(g) = \pm 1$, group $P$ acts on $V$ by scalar matrices, and hence, any subspace is invariant under action of $P$. Space $V$ has the following basis $1 \otimes v$, $x \otimes v$, where $v$ is a basis of one-dimensional ${\mathbb C}P$ - module ${\mathbb C}^{\chi}$. One can show that ${\mathbb C}^{\chi}_1$ and ${\mathbb C}^{\chi}_{-1}$ are one-dimensional subspaces generated by $1 \otimes v + x \otimes v$ and $1 \otimes v - x \otimes v$ respectively.
Actually, $x(1 \otimes v + x \otimes v) = x(1+x) \otimes v = (x+x^2) \otimes v = (1+x) \otimes v$ and $x(1 \otimes v - x \otimes v) = x(1-x) \otimes v = -(1-x)\otimes v$.
\end{proof}

\section{Space ${\cal L}_{\theta}$, representations of $P$ and deformations of Klein group.}

Let us come back to subspace ${\cal L}_{\theta} \subset M_4({\mathbb C})$ generated by matrices $X,Y,Z$.
\begin{Proposition}
\label{theor}
\begin{itemize}
\item{
For fixed $\theta \in {\mathbb C}^*, \theta \ne \pm 1$ representation $\phi$ of group $G = ({\mathbb Z} \oplus {\mathbb Z}_2) \rtimes {\mathbb Z}_2 = P \rtimes {\mathbb Z}_2$ defined by matrices $X,Y,Z$ is isomorphic to ${\mathbb C}G \otimes_{{\mathbb C}P}{\mathbb C}^{\chi_1} \oplus {\mathbb C}G \otimes_{{\mathbb C}P}{\mathbb C}^{\chi_2}$, where  $\chi_1 = (\chi_1(g) = \theta, \chi(z) = 1)$ and $\chi_2 = (\chi_2(g) = -\theta, \chi_2(z) = -1)$.
}
\item{
In the case $\theta = \pm 1$ matrices $X,Y,Z$ define representation of $G$: ${\mathbb C}^{\chi_1}_{1} \oplus {\mathbb C}^{\chi_1}_{-1} \oplus {\mathbb C}^{\chi_2}_{1} \oplus {\mathbb C}^{\chi_2}_{-1}$, where $\chi_1 = (\chi_1(g) = \theta, \chi(z) = 1)$ and $\chi_2 = (\chi_2(g) = -\theta, \chi_2(z) = -1)$.
}
\end{itemize}
\end{Proposition}
\begin{proof}
Using described early technics, representation of ${\mathbb C}G$ is given by representation of ${\mathbb C}P$. Consider matrices $XY$ and $Z$.
\begin{equation}
XY =
\begin{pmatrix}
0 & 0 & 0 & 1/\theta\\
0 & 0 & \theta & 0\\
0 & \theta & 0 & 0\\
1/\theta & 0 & 0 & 0
\end{pmatrix},
Z =
\begin{pmatrix}
0 & 0 & 0 & 1\\
0 & 0 & 1 & 0\\
0 & 1 & 0 & 0\\
1 & 0 & 0 & 0
\end{pmatrix}
\end{equation}
One can show that there is a basis in which matrix $XY$ and $Z$ have the following view:
\begin{equation}
XY =
\begin{pmatrix}
1/\theta & 0 & 0 & 0\\
0 & \theta & 0 & 0\\
0 & 0 & -1/\theta & 0\\
0 & 0 & 0 & -\theta
\end{pmatrix},
Z =
\begin{pmatrix}
1 & 0 & 0 & 0\\
0 & 1 & 0 & 0\\
0 & 0 & -1 & 0\\
0 & 0 & 0 & -1
\end{pmatrix}
\end{equation}

Let $V$ be a 4-dimensional ${\mathbb C}G$ - module given by matrices $X,Y,Z$.
If $\theta \ne \pm 1$, then we have two irreducible submodules ${\mathbb C}G \otimes_{{\mathbb C}P} {\mathbb C}^{\chi_1}$ and ${\mathbb C}G \otimes_{{\mathbb C}P} {\mathbb C}^{\chi_2}$ of $M$, where $\chi_1 = (\chi_1(g) = \theta, \chi_1(z) = 1)$ and $\chi_2 = (\chi_2(g) = -\theta, \chi_2(z) = -1)$.

Therefore, we have the following commutative diagram:
\begin{equation}
\xymatrix{
0 \ar[r] & {\mathbb C}G \otimes_{{\mathbb C}P} {\mathbb C}^{\chi_1} \ar[r] & V \ar[r] & W \ar[r] & 0\\
& & {\mathbb C}G \otimes_{{\mathbb C}P} {\mathbb C}^{\chi_2}\ar[u]
}
\end{equation}
Since ${\mathbb C}G \otimes_{{\mathbb C}P} {\mathbb C}^{\chi_1}$ and ${\mathbb C}G \otimes_{{\mathbb C}P} {\mathbb C}^{\chi_2}$ are not isomorphic and irreducible, then composition ${\mathbb C}G \otimes_{{\mathbb C}P} {\mathbb C}^{\chi_2} \to V \to W$ is an immersion. Also, note that ${\rm dim}_{\mathbb C}W = {\rm dim}_{\mathbb C}{\mathbb C}G \otimes_{{\mathbb C}P} {\mathbb C}^{\chi_2} = 2$. Thus, ${\mathbb C}G \otimes_{{\mathbb C}P} {\mathbb C}^{\chi_2} \cong W$. Thus, ${\mathbb C}G$-module $V$ is a direct sum ${\mathbb C}G \otimes_{{\mathbb C}P} {\mathbb C}^{\chi_1} \oplus {\mathbb C}G \otimes_{{\mathbb C}P} {\mathbb C}^{\chi_2}$.

As we know from section \ref{kleingr} if $\theta = \pm 1$, then representation $\phi$ defines regular representation of Klein group.
\end{proof}

For further studying of representations of $G$ from proposition, consider ${\mathbb C}G$ - module $V$ corresponding to representation of $G$ given by matrices $X,Y,Z$. Using theorem, we get that ${\rm Char}(V) = \{ (\chi(g) = \theta, \chi(z) = 1), (\chi(g) = 1/\theta, \chi(z) = 1), (\chi(g) = -\theta, \chi(z) = -1), (\chi(g) = -1/\theta, \chi(z) = -1)\}$. One can define the automorphism $s$ of group algebra (not group!) ${\mathbb C}P$ defined by formula:
\begin{equation}
g \mapsto -g, z \mapsto -z.
\end{equation}
Also, we have the involution $c_x: (g,z) \mapsto (g^{-1},z)$ defined early.
One can check that automorphisms $s$ and $c_x$ commute. Thus, we have an action of ${\mathbb Z}_2 \times {\mathbb Z}_2$ on group algebra ${\mathbb C}P$.
Consider algebra of invariants ${\mathbb C}P^{{\mathbb Z}_2 \times {\mathbb Z}_2} = \{ a \in {\mathbb C}P| c_x(a) = a, s(a) = a\}$.
\begin{Proposition}
Algebra of invariants ${\mathbb C}P^{{\mathbb Z}_2 \times {\mathbb Z}_2}$ is generated by elements $v = (g+g^{-1})z$. Also, ${\mathbb C}P$ is a free ${\mathbb C}P^{{\mathbb Z}_2 \times {\mathbb Z}_2}$ - module of rank 4.
\end{Proposition}
\begin{proof}
It can be shown that ${\mathbb C}P^{{\mathbb Z}_2 \times {\mathbb Z}_2} = {\cal C}^{s} = \{p \in {\cal C}| s(p) = p\}$. Recall that ${\cal C}$ is a quotient of ${\mathbb C}[u,z]$ by relation $z^2 = 1$.
It is easy that relation $z^2 = 1$ is $s$-invariant. Thus, ${\cal C}^{s}$ is a quotient of ${\mathbb C}[u,z]^{s} = \{p(u,z) = p(-u,-z)\}$ by relation $z^2 = 1$. It can be calculated that ${\mathbb C}[u,z]^{s} = {\mathbb C}[v_1 = u^2,v = uz,v_3 = z^2, v_1v_3 = v^2_2]$. Using relation $z^2 = 1$, we get that ${\cal C}^{s} = {\mathbb C}[v_1, v, v_1 = v^2]$.
Analogous to proposition \ref{freemod}, one can prove that ${\mathbb C}P$ is a free ${\mathbb C}P^{{\mathbb Z}_2 \times {\mathbb Z}_2}$ - module of rank 4.
\end{proof}
\begin{Corollary}
\label{freeeight}
${\mathbb C}G$ is a free ${\mathbb C}[v]$ - module of rank 8.
\end{Corollary}

Fix character $\chi$ of ${\cal C}^s$. Denote by $b$ the value $\chi(v)$. Denote by ${\mathbb C}^{\chi}$ the ${\mathbb C}[v]$ - module corresponding to $\chi$. Also, we can consider the following algebra
${\cal A}_{\theta} = {\mathbb C}G \otimes_{{\cal C}^s} {\mathbb C}^{\chi}$. For fixed $b \in {\mathbb C}$ algebra ${\cal A}_{\theta}$ is a quotient of ${\mathbb C}G$ by relation:
\begin{equation}
\label{relab}
(g+g^{-1})z = b\cdot 1,
\end{equation}
here $b = \chi(z(g+g^{-1})) = \theta+\theta^{-1}$.
Consider representation $\phi$ of $G$. We have the following commutative diagram:
\begin{equation}
\xymatrix{{\mathbb C}G \ar[rr]^{\phi}\ar[rd] && Mat_4({\mathbb C})\\
& {\cal A}_{\theta} \ar[ru]
}
\end{equation}
It means that for any $\theta \in {\mathbb C}^*$ matrix $\phi(g+g^{-1})z = (XY + YX)Z$ is $(\theta + \theta^{-1})I$, where $I$ is an identity matrix.

Let us prove the following
\begin{Theorem}\label{theorem}
For fixed $\theta \in {\mathbb C}^*$ ${\rm dim}_{\mathbb C}{\cal A}_{\theta} = 8$.
\begin{itemize}
\item{
If $b=\theta + \theta^{-1} \ne 0$ then algebra ${\cal A}_{\theta}$ has the following basis $1,g,g^2,g^3,x,xg,xg^2,xg^3$ and the relations:
\begin{equation}
x^2 = 1, g^4 + (2-b^2)g^2 + 1 = 0, xgx = g^{-1} = (b^2 - 2)g - g^3, b = \theta+\theta^{-1}.
\end{equation}
Also, we have canonical isomorphism: ${\cal A}_{\theta} \cong {\cal A}_{-\theta}$ for any $\theta \in {\mathbb C}, \theta \ne \pm {\rm i}$.
}
\item{
If $b=\theta+\theta^{-1} = 0$, then ${\cal A}_{\pm {\rm i}}$ has the following basis $1,g,x,z,xg,xz,gz,xgz$ and the relations:
\begin{equation}
x^2 = z^2 = 1, g^2 = -1, xgx = g^{-1}, xz = zx, gz = zg
\end{equation}
}
\end{itemize}
Also, we have the following description of algebra ${\cal A}_{\theta}$ for various ${\theta} \in {\mathbb C}^*$:
\begin{itemize}
\item{
If $\theta \ne \pm 1$, then ${\cal A}_{\theta} \cong {\rm Mat}_2({\mathbb C}) \oplus {\rm Mat}_2({\mathbb C})$.
}
\item{
If $\theta = \pm 1$, then algebra ${\cal A}_{\pm 1}$ has 4-dimensional Jacobson radical $J$ and ${\cal A}_{\pm 1}/J = \bigoplus^{4}_{i=1}{\mathbb C}$.
}
\end{itemize}
\end{Theorem}
\begin{proof}
Using corollary \ref{freeeight}, we get that for fixed $\theta \in {\mathbb C}$ algebra ${\cal A}_{\theta}$ has dimension 8. Using relation $(g+g^{-1})z = b \cdot 1$, we get that $g+g^{-1} = b \cdot z$. If $b =\ne 0$, then $z = \frac{1}{b}(g + g^{-1})$. Also, $b^2 \cdot z^2 = b^2 \cdot 1 = (g + g^{-1})^2$. Thus, we can express $z$ in terms of $g$ and since ${\rm dim}_{\mathbb C}{\cal A}_{\theta} = 8$, we obtain that ${\cal A}_{\theta }$ has the required basis and the relations.
If $b= 0$, then $(g+g^{-1})z = 0$. Since $z^2 = 1$, we get that $g + g^{-1} = 0$. Also, we get that algebra ${\cal A}_{\pm {\rm i}}$ has the required basis and the relations.
Existence of canonical isomorphism is trivial.

For studying algebra ${\cal A}_{\theta}$, we have to study center of algebra ${\cal A}_{\theta}$. If $b \ne 0$ then this center has basis $1, g+g^{-1}$ with relation $(g+g^{-1})^2 = b^2 \cdot 1$. Further, center has two characters $g+g^{-1} \mapsto b = \theta+\theta^{-1}$ and $g+g^{-1} \mapsto -b = -(\theta+\theta^{-1})$.
If $\theta = \pm {\rm i}$, then center has a basis $1,z$ and relation $z^2 = 1$. Thus, there are two characters of center. Analogous to proposition \ref{algirr}, one can show that if $\theta \ne \pm 1$ algebra ${\cal A}_{\theta}$ is a direct sum of matrix algebras.

Using canonical isomorphism ${\cal A}_{\theta}$ and ${\cal A}_{-\theta}$, we get that it is sufficient to consider the case of algebra ${\cal A}_1$.
Using proposition \ref{algirr}, we get that if $\theta = 1$, algebras ${\cal A}_{1}$ has 4-dimensional Jacobson radical. Actually, in this case $(g^2 - 1)^2 = 0$. Consider ideal $J$ of ${\cal A}_{1}$ generated by element $g^2 - 1$. Thus, ideal $J$ has a basis $(g^2-1),x(g^2-1),g(g^2-1),xg(g^2-1)$. One can check that $J^2 = 0$. Therefore, algebras ${\cal A}_{\pm 1}$ has 4-dimensional Jacobson radical.
\end{proof}
Further, let us consider image of algebras ${\cal A}_{\theta}$ under representation $\phi$ defined by matrices $X,Y,Z$. Remind the following notion: representation of algebra is called {\it semisimple} iff this representation is a direct sum of irreducible ones.
Recall the following property of Jacobson radical: image of radical under any semisimple representation is zero.
Subalgebra ${\cal L}_{\theta} \subset Mat_4({\mathbb C})$ is an image of algebra ${\cal A}_{\theta}$ for arbitrary $\theta$. Using canonical isomorphism ${\cal A}_{\theta} \cong {\cal A}_{-\theta}$, we get that $b$ is defined up to sign. One can find that we can choose $b = \chi(z(g+g^{-1})) = \pm (\theta + \theta^{-1})$.
Therefore, we have the following
\begin{Corollary}\label{main}
Consider subalgebra ${\cal M}_{\theta} \subset Mat_4({\mathbb C})$ generated by ${\cal L}_{\theta} \subset Mat_4({\mathbb C})$. We have the following possibilities:
\begin{itemize}
\item{
If $\theta \ne \pm 1$, then this ${\cal M}_{\theta}$ is a direct sum $Mat_2({\mathbb C}) \oplus Mat_2({\mathbb C})$.
}
\item{
Assume that $\theta = \pm 1$. In this case subalgebra ${\cal M}_{\theta}$ is an image of Klein group or direct sum of four ${\mathbb C}$'s.
}
\end{itemize}
\end{Corollary}

\section*{Appendix A. Homological properties of irreducible ${\mathbb C}G$-modules.}

In this section we will introduce some notions of homological algebra and prove some technical facts on some ${\mathbb C}G$-modules.
One can find classical definitions and results in homological algebra in many books, for example \cite{Maclane}, \cite{GelfMan}, \cite{Eil}

Recall the notion of extension group. Fix some associative algebra $A$. Suppose that we have two $A$-modules $W_1$ and $W_2$. We would like to classify $A$-modules $W$ satisfying the following conditions: $W_1$ is a submodule of $W$ and quotient $W/W_1$ is isomorphic to $W_2$. For this purpose we will introduce the following equivalence relation: we will say that $W$ is equivalent to $W'$ iff there is a following commutative diagram of $A$-modules:
\begin{equation}
\xymatrix{
0\ar[r] & W_1\ar[r] \ar[d] & W \ar[r]\ar[d] & W_2\ar[r]\ar[d] & 0\\
0\ar[r] & W_1\ar[r] & W'\ar[r] & W_2 \ar[r] & 0
}
\end{equation}
where all morphisms are isomorphisms. Standard arguments from homological algebra tell us that set of equivalence classes is an abelian group. This group is called extension group ${\rm Ext}^1_{A}(W_2,W_1)$. Neutral element of this group is a direct sum $W_1 \oplus W_2$.

Extension group has the following description in the case of commutative algebras. Assume that $A$ is a finite generated commutative algebra. Consider algebraic variety ${\rm Spec}A = {\rm Hom}_{alg}(A,{\mathbb C})$. Let $\chi$ be a point of ${\rm Spec}A$. ${\mathbb C}^{\chi}$ is a corresponding $A$-module. Then ${\rm Ext}^1_{A}({\mathbb C}^{\chi}, {\mathbb C}^{\chi})$ is a tangent space of ${\rm Spec}A$ at the point $\chi$. Recall the following way of calculation of extension group ${\rm Ext}^1_A(V_1,V_2)$ of $A$-modules $V_1$ and $V_2$ for arbitrary algebra $A$. Assume that we  have a projective resolution of $A$-module $V_1$:
\begin{equation}
\xymatrix{
0 \ar[r] & P_s\ar[r] & ... \ar[r] & P_1\ar[r] & V_1\ar[r] & 0,
}
\end{equation}
where $P_i, i = 1,...,s$ are projective $A$-modules. Applying to this resolution functor ${\rm Hom}_A(-,V_2)$, we get the following complex (not exact sequence!):
\begin{equation}
\label{compl}
\xymatrix{
0 \ar[r] & {\rm Hom}_A(P_1,V_2) \ar[r] & ... \ar[r] & {\rm Hom}_A(P_s,V_2) \ar[r] & 0.
}
\end{equation}
Zeroth cohomology group of this complex is a group ${\rm Hom}_A(V_1,V_2)$, first cohomology group is a group ${\rm Ext}^1_A(V_1,V_2)$. Note that there is a notion of extension group ${\rm Ext}^i_A(V_1,V_2)$ of $A$-modules $V_1$ and $V_2$. In this case ${\rm Ext}^i_A(V_1,V_2)$ is a ith cohomology group of the complex (\ref{compl}).

Consider one-dimensional ${\mathbb C}P$-module ${\mathbb C}^{\chi}$, where $\chi$ is a character of ${\mathbb C}[g^{\pm 1}] \otimes {\mathbb C}[z, z^2 = 1]$.
Restrict ${\mathbb C}P$ - module ${\mathbb C}^{\chi}$ to algebra of Laurent polynomials ${\mathbb C}[g^{\pm 1}]$.
It can be shown in usual way that there is a free resolution (and, hence, projective resolution) of ${\mathbb C}^{\chi}$ as ${\mathbb C}[g^{\pm 1}]$ - module:
\begin{equation}
\label{res}
\xymatrix{
0\ar[r] & {\mathbb C}[g^{\pm 1}]\ar[r]^{j} & {\mathbb C}[g^{\pm 1}]\ar[r]^{p} & {\mathbb C}^{\chi}\ar[r] & 0.
}
\end{equation}
Actually, ${\mathbb C}^{\chi}$ is a quotient of ${\mathbb C}[g^{\pm 1}]$ by ideal generated by $g - \chi(g)1$. Morphism $p$ is a natural projection:
$p: {\mathbb C}[g^{\pm 1}] \to {\mathbb C}[g^{\pm 1}]/<g - \chi(g)1>$. Morphism $j$ is a morphism ${\mathbb C}[g^{\pm 1}] \to {\mathbb C}[g^{\pm 1}]$ is generated by rule $j(1) = (g-\chi(g)\cdot 1)$ and, hence, $j(g) = g(g - \chi(g)\cdot 1)$. One can prove that this morphism is injective. Also, image $j({\mathbb C}[g^{\pm 1}])$ coincides with ${\rm Ker}p$.
Further, come back to ${\mathbb C}P$-module ${\mathbb C}^{\chi}$. In this case ${\mathbb C}^{\chi}$ is a quotient of ${\mathbb C}P$ by ideal generated by elements $g - \chi(g)1, z - \chi(z)1$. It follows that $(z+1)^2 = 2(z+1)$ and $(1-z)^2 = 2(1-z)$. We have the following decomposition ${\mathbb C}P$ into direct sum of projective ${\mathbb C}P$-modules: ${\mathbb C}P = {\mathbb C}P(z+1) \oplus {\mathbb C}P(z-1)$. It can be shown that there is a projective resolution of ${\mathbb C}^{\chi}$ as ${\mathbb C}P$-module:
\begin{equation}
\label{resgf}
\xymatrix{
0\ar[r] & {\mathbb C}P(z\pm 1)\ar[r]^{j} & {\mathbb C}P(z\pm 1)\ar[r]^{p} & {\mathbb C}^{\chi}\ar[r] & 0,
}
\end{equation}
where we take plus or minus simultaneously in both modules as follows: if $\chi(z) = 1$ then we take plus, else we take minus.

Using standard arguments and ${\mathbb C}P$ - resolution of ${\mathbb C}^{\chi}$, one can formulate the following proposition:
\begin{Proposition}
Consider two points $\chi, \psi \in {\rm Spec}{\mathbb C}P$. Then we have the following statements:
\begin{itemize}
\item{
If $\chi \ne \psi$, then ${\rm Ext}^1_{{\mathbb C}P}({\mathbb C}^{\chi},{\mathbb C}^{\psi}) = 0$
}
\item{
${\rm Ext}^1_{{\mathbb C}P}({\mathbb C}^{\chi}, {\mathbb C}^{\chi}) = {\mathbb C}$.
}
\end{itemize}
\end{Proposition}
\begin{proof}
Applying to resolution (\ref{resgf}) functor ${\rm Hom}_{{\mathbb C}[g,g^{-1}]}(-,{\mathbb C}^{\psi})$, we get the following complex:
$$
{\rm Hom}_{{\mathbb C}P}({\mathbb C}P(z \pm 1),{\mathbb C}^{\psi}) \to {\rm Hom}_{{\mathbb C}P}({\mathbb C}P(z \pm 1),{\mathbb C}^{\psi}).
$$
One can show that if $\chi \ne \psi$ then this map is isomorphism, and hence, zeroth and first cohomology groups are trivial. In the case $\chi = \psi$ we obtain that ${\rm Hom}_{{\mathbb C}P}({\mathbb C}^{\chi}, {\mathbb C}^{\chi}) = {\rm Ext}^1_{{\mathbb C}P}({\mathbb C}^{\chi}, {\mathbb C}^{\chi}) = {\mathbb C}$.
\end{proof}

Also, we can make some remarks on homological properties of so-called induced ${\mathbb C}G$-modules.

{\bf Remark.}
Using resolution (\ref{resgf}), we can get resolution of some ${\mathbb C}G$-modules.
Tensoring sequence (\ref{resgf}) by ${\mathbb C}G$ over ${\mathbb C}P$, we get the following sequence:
\begin{equation}
\label{rescg}
\xymatrix{
0\ar[r] & {\mathbb C}G \otimes_{{\mathbb C}P} {\mathbb C}P(z \pm 1) \ar[r] & {\mathbb C}G \otimes_{{\mathbb C}P} {\mathbb C}P(z \pm 1) \ar[r] & {\mathbb C}G \otimes_{{\mathbb C}P} {\mathbb C}^{\chi} \ar[r] & 0.
}
\end{equation}
Module ${\mathbb C}G \otimes_{{\mathbb C}P} {\mathbb C}^{\chi}$ is called module {\it induced} by character $\chi$ of subgroup $P$.
Sequence (\ref{rescg}) is a projective resolution of ${\mathbb C}G \otimes_{{\mathbb C}P} {\mathbb C}^{\chi}$, i.e. ${\mathbb C}G \otimes_{{\mathbb C}P} {\mathbb C}P(z \pm 1)$ is a projective ${\mathbb C}G$ - module.

Applying standard arguments to ${\mathbb C}G$-modules $W_{\chi} = {\mathbb C}G \otimes_{{\mathbb C}P} {\mathbb C}^{\chi}$, $W_{\psi} = {\mathbb C}G \otimes_{{\mathbb C}P} {\mathbb C}^{\psi}$ and resolution (\ref{rescg}). We get the following proposition:
\begin{Proposition}
We have the following isomorphisms for $W_{\chi}$ and $W_{\psi}$:
\begin{itemize}
\item{
if $\chi \ne \psi, \psi^{-1}$, then
\begin{equation}
{\rm Ext}^1_{{\mathbb C}G}(W_{\chi},W_{\psi}) = 0.
\end{equation}
}
\item{
if $\chi = \psi, \psi^{-1} \ne \pm 1$
\begin{equation}
{\rm Ext}^1_{{\mathbb C}G}(W_{\chi}, W_{\psi}) = {\mathbb C},
\end{equation}
}
\item{
if $\chi = \psi = \pm 1$, then
\begin{equation}
{\rm Ext}^1_{{\mathbb C}G}(W_{\chi}, W_{\psi}) = {\mathbb C}^2,
\end{equation}
}
\end{itemize}
\end{Proposition}
\begin{proof}
we have the following complex:
\begin{equation}
\xymatrix{
 {\rm Hom}_{{\mathbb C}G}({\mathbb C}G \otimes_{{\mathbb C}P} {\mathbb C}P(z \pm 1), W_{\psi}) \ar[r] &
{\rm Hom}_{{\mathbb C}G}({\mathbb C}G \otimes_{{\mathbb C}P} {\mathbb C}P(z \pm 1), W_{\psi}).
}
\end{equation}
Zeroth and first cohomology group of this complex are isomorphic, i.e. ${\rm Hom}_{{\mathbb C}G}(W_{\chi}, W_{\psi}) \cong {\rm Ext}^1_{{\mathbb C}G}(W_{\chi},W_{\psi})$.
Using adjacency of functors and corollary \ref{isomor}, we get that
$$
{\rm Hom}_{{\mathbb C}G}(W_{\chi},W_{\psi}) \cong {\rm Hom}_{{\mathbb C}P}({\mathbb C}^{\chi}, {\mathbb C}^{\psi} \oplus {\mathbb C}^{c_x(\psi)}).
$$
Using this isomorphism, we obtain the statement of the proposition.
\end{proof}

\section*{Appendix B. Some aspects of noncommutative geometry in representation theory of $G$.}

We can describe the connection between ${\mathbb C}G$ - modules and ${\mathbb C}P$ - modules from point of view of noncommutative algebraic geometry also. One can find many aspects of noncommutative algebraic geometry in many papers, for example, \cite{Kraft}, \cite{KonRos}, \cite{Proc}.

Consider a finite-generated associative algebra $A$. Recall that {\it a variety of representations} of algebra $A$ is called the variety ${\bf Rep}_n(A) = {\rm Hom}_{alg}(A, Mat_n({\mathbb C}))$. Also, there is a natural action of ${\rm GL}_n({\mathbb C})$ on $Mat_n({\mathbb C})$ by conjugation. Equivalently, if we fix basis in n-dimensional space ${\mathbb C}^n$, then we have an isomorphism: $Mat_n({\mathbb C}) \cong {\rm End}_{\mathbb C}({\mathbb C}^n)$. In this case, the group ${\rm GL}_n({\mathbb C})$ acts on $Mat_n({\mathbb C})$ by substitutions of bases in ${\mathbb C}^n$. Using this action, we have well-defined action of ${\rm GL}_n({\mathbb C})$ on ${\bf Rep}_n(A)$. Thus, we can consider the quotient of ${\bf Rep}_n(A)$ by action of ${\rm GL}_n({\mathbb C})$. This quotient is called by {\it a moduli variety} of algebra $A$. We will denote this variety by ${\cal M}_n(A)$ for an arbitrary algebra $A$.

We can consider any element $a \in A$ as a matrix-valued function on ${\bf Rep}_n(A)$. Namely, $a(\rho) = \rho(a), \rho \in {\bf Rep}_n(A)$. Also, we can define ${\rm GL}_n({\mathbb C})$ - invariant functions on ${\bf Rep}_n(A)$ in the following manner: ${\rm Tr}(a) (\rho) = {\rm Tr}\rho(a), \rho \in {\bf Rep}_n(A)$. Using ${\rm GL}_n({\mathbb C})$ - invariance, we can consider these functions as functions on ${\cal M}_n(A)$.

Recall the following well-known result.
\begin{Proposition}{\cite{Kraft}}
\label{genmod}
For a finite-generated algebra $A$, the ring of regular polynomial functions on ${\cal M}_n(A)$ is generated by functions ${\rm Tr}a, a \in A$.
\end{Proposition}
We will say that $A$-module $V$ has {\it Jordan-Holder composite factors} (briefly, composite factors) $W_1$,...,$W_s$ if there is a sequence of $A$-submodules of the following type:
\begin{equation}
0 = M_0 \subset M_1 \subset M_2 ... \subset M_{s-1} \subset M_s = V
\end{equation}
and the quotients $W_i = M_i/M_{i-1}, i = 1,...,s$ are irreducible $A$-modules. It is known that the set of composite factors for any $A$-module is unique up to permutation.
Let us denote ${\bf gr}(V)$ the direct sum $W_1 \oplus ... \oplus W_s$ for $A$-module $V$ with composite factors $W_1$,...,$W_s$.
Two $n$-dimensional $A$-modules $V_1$ and $V_2$ correspond to the same point of ${\cal M}_n(A)$ iff ${\bf gr}(V_1) \cong {\bf gr}(V_2)$. Equivalently, $V_1$ and $V_2$ correspond to the same point of ${\cal M}_n(A)$ iff traces of any elements $a \in A$ on $V_1$ and $V_2$ are the same, i.e. ${\rm Tr}a |_{V_1} = {\rm Tr}a |_{V_2}$ for any element $a \in A$.

Let us come back to our situation.
We have an immersion of algebras: ${\mathbb C}P \to {\mathbb C}G$. Thus, we have induced maps:
\begin{equation}
p_1: {\bf Rep}_2({\mathbb C}G) \to {\bf Rep}_2 ({\mathbb C}P)
\end{equation}
and
\begin{equation}
p_2: {\cal M}_2({\mathbb C}G) \to {\cal M}_2 ({\mathbb C}P).
\end{equation}

Let us describe variety ${\bf Rep}_2{\mathbb C}P$ as follows. Up to ${\rm GL}_2({\mathbb C})$ - conjugacy, we have three possibilities for picking up matrix corresponding to element $g$:
\begin{equation}
1.
\begin{pmatrix}
\alpha & 0\\
0 & \beta
\end{pmatrix},
2.
\begin{pmatrix}
\alpha & 1\\
0 & \alpha
\end{pmatrix},
3.
\begin{pmatrix}
\alpha & 0\\
0 & \alpha
\end{pmatrix},
\alpha, \beta \in {\mathbb C}^*.
\end{equation}
Recall that $z$ commutes with $g$. Thus, we have the following possibilities for $z$ in the first and third cases:
\begin{equation}
\begin{pmatrix}
1 & 0\\
0 & 1
\end{pmatrix},
\begin{pmatrix}
-1 & 0\\
0 & -1
\end{pmatrix},
\begin{pmatrix}
1 & 0\\
0 & -1
\end{pmatrix}.
\end{equation}
In the second case, we have the following situation for element $z$:
\begin{equation}
\begin{pmatrix}
1 & 0\\
0 & 1
\end{pmatrix},
\begin{pmatrix}
-1 & 0\\
0 & -1
\end{pmatrix}.
\end{equation}
Further, consider ${\cal M}_2({\mathbb C}P)$. It can be shown that the ring of regular polynomial functions on ${\cal M}_2({\mathbb C}P)$ is generated by ${\rm Tr}(g), {\rm Tr}(g^2), {\rm Tr}(z)$. Also, function ${\rm Tr}(z)$ can have only three possible values $\{-2, 0, 2\}$. It means that every point of ${\cal M}_2({\mathbb C}P)$ is defined by two continuous parameters ${\rm Tr}(g), {\rm Tr}(g^2)$ and discrete parameter ${\rm Tr}(z)$.
We have the following proposition: 
\begin{Proposition}
Variety ${\cal M}_2({\mathbb C}P)$ has three irreducible components: $U_{-}$, $U_{0}$ and $U_{+}$. Every component is Zarisski - open subset of affine plane ${\mathbb C}^2$ given by relation: $({\rm Tr}g)^2 - {\rm Tr}g^2 \ne 0$. These components are indexed by values ${\rm Tr}(z) = -2, 0, 2$ respectively.
\end{Proposition}
\begin{proof}
Recall that we have the following relation for $2 \times 2$ matrices:
\begin{equation}
det(g) = \frac12 (({\rm Tr}g)^2 - {\rm Tr}g^2).
\end{equation}
Since $g$ is invertible, $det(g) \ne 0$. Thus, we get the required statement.
\end{proof}

Consider variety ${\bf Rep}_2{\mathbb C}G$. We have the following cases for matrices corresponding to $x, y, z$:
\begin{itemize}
\item{
$x,y,z$ are reflections,}
\item{
$z$ is scalar,
}
\end{itemize}
Consider first case. In this case $x$, $y$ and $z$ are reflections. Recall that $z$ commutes with $x$ and $y$.
Note the following proposition:
\begin{Proposition}
If two reflections $x$ and $z$ commute, then there is a basis in which $x$ and $z$ are diagonal matrices.
\end{Proposition}
\begin{proof}
Pick up eigenvectors $v_1,v_2$ of $z$ such that $zv_1 = v_1, zv_2 = -v_2$. We have the following relation:
$x z v_1 = x v_1$ and $zx v_1 = x v_1$. Hence, $x v_1$ is eigenvector corresponding to eigenvalue $1$. Therefore, $x v_1 = \alpha v_1$ for some $\alpha \in {\mathbb C}^*$. Since $x$ is a reflection, then $\alpha = \pm 1$. The rest is analogous.
\end{proof}

Thus, if $x,y,z$ are reflections, then $x,y,z$ commuting operators. Thus, $x,y,z$ have the following view:
\begin{equation}
\label{iden}
\begin{pmatrix}
\pm 1 & 0\\
0 & \pm 1
\end{pmatrix}
\end{equation}
Consider second case. If $x$ and $y$ correspond to commuting matrices, then this case is quite similar to first case. If $x$ and $y$ are not commuting then matrix of $z$ is scalar.

\begin{Proposition}
We have the following decomposition:
\begin{equation}
{\cal M}_2 ({\mathbb C}G) = {\mathbb C}^1_+ \cup {\mathbb C}^1_{-} \cup S,
\end{equation}
where ${\mathbb C}^1_+$, ${\mathbb C}^1_{-}$ are 1-dimensional families of ${\mathbb C}G$ - modules corresponding to ${\rm Tr}z = \pm 2 {\rm Tr}x = {\rm Tr}y = 0$ respectively, $S$ is a set of isolated points, $|S| = 27$.
\end{Proposition}
\begin{proof}
Using proposition \ref{genmod}, we get that the ring of regular polynomial functions on ${\cal M}_2 ({\mathbb C}G)$ has the following generators: ${\rm Tr}(x), {\rm Tr}(y), {\rm Tr}(z), {\rm Tr}(xy), {\rm Tr}(xz), {\rm Tr}(yz)$.
We have the following cases:
\begin{itemize}
\item{$x,y,z$ are reflections and correspond to matrices of type (\ref{iden}). In this case we have ${\rm Tr}z = {\rm Tr}x = {\rm Tr}y = 0$ and $({\rm Tr}xy, {\rm Tr}xz, {\rm Tr}yz) = (2,-2,-2), (-2,2,-2), (-2,-2,2)$. In this case there are 3 points in $S$.}
\item{$y,z$ are reflections, $x$ is a scalar matrix. In this case $y$ and $z$ are matrices of type (\ref{iden}). Thus, ${\rm Tr}z = {\rm Tr}y = 0$, ${\rm Tr}x = \pm 2$, ${\rm Tr}xy = {\rm Tr}xz = 0$ and ${\rm Tr}yz = \pm 2$. In this case there are 2 points in $S$.}
\item{$x,z$ are reflections, $y$ is a scalar matrix. This case is analogous to second one. Thus, ${\rm Tr}z = {\rm Tr}x = 0$, ${\rm Tr}y = \pm 2$, ${\rm Tr}xy = {\rm Tr}yz = 0$ and ${\rm Tr}xz = \pm 2$. In this case there are 2 points in $S$.}
\item{$x$ is a reflection, $y,z$ are scalar matrices. In this case, ${\rm Tr}x = 0, {\rm Tr}y = \pm 2, {\rm Tr}z = \pm 2$ and $({\rm Tr}xy, {\rm Tr}xz, {\rm Tr}yz) = (0,0,2), (0,0,-2)$. In this case there are 4 points in $S$.}
\item{$y$ is a reflection, $x,z$ are scalar matrices. In this case, ${\rm Tr}y = 0, {\rm Tr}x = \pm 2, {\rm Tr}z = \pm 2$ and $({\rm Tr}xy, {\rm Tr}xz, {\rm Tr}yz) = (0,2,0), (0,-2,0)$. In this case there are 4 points in $S$.}
\item{$x,y,z$ are scalar matrices. In this case ${\rm Tr}x = \pm 2, {\rm Tr}y = \pm 2, {\rm Tr}z = \pm 2$, $({\rm Tr}xy, {\rm Tr}xz, {\rm Tr}yz) = (2,2,2), (2,-2,-2), (-2,2,-2), (-2,-2,2)$. In this case there are 8 points in $S$.}
\item{$z$ is reflection, $x,y$ are scalar matrices. In this case ${\rm Tr}z = 0, {\rm Tr}x = \pm 2, {\rm Tr}y = \pm 2$ and ${\rm Tr}xy = \pm 2, {\rm Tr}xz = 0, {\rm Tr}yz = 0$ In this case there are 4 points in $S$.}
\item{$z$ is scalar and $x,y$ are reflections. In this case ${\rm Tr}z = \pm 2$ and ${\rm Tr}x = {\rm Tr}y = 0$. Also, ${\rm Tr}xz = {\rm Tr}yz = 0$. In this case ${\mathbb C}G$ - modules are parameterized by ${\rm Tr}xy$. In this case we have two 1-dimensional families ${\mathbb C}G$-modules. We will denote two components corresponding to ${\rm Tr}z = 2, {\rm Tr}x = {\rm Tr}y = 0$ and ${\rm Tr}z = -2, {\rm Tr}x = {\rm Tr}y = 0$ by ${\mathbb C}^1_{+}$ and ${\mathbb C}^1_{-}$ respectively.}
\end{itemize}
\end{proof}

Let us come back to map $p_2$.
\begin{Proposition}
Varieties $p_2({\mathbb C}^1_{-})$ and $p_2({\mathbb C}^1_{+})$ are curves given by equations:
\begin{equation}
{{\rm Tr}^2 (g) - {\rm Tr}(g^2)} = 2.
\end{equation}
in $U_{-}$ and $U_{+}$ respectively.
\end{Proposition}
\begin{proof}
One can check that $p_2({\mathbb C}^1_{-}) \subset U_{-}$ and $p_2({\mathbb C}^1_{+}) \subset U_{+}$. Also, remind that $g = xy$. Thus, in the 1-dimensional components, we have the following identity: ${\rm det}g^2 = {\rm det}xyxy = {\rm det}^2 x \cdot {\rm det}^2 y$. Since $x$ and $y$ are reflections< we get that ${\rm det}g^2 = 1$. Using Gamilton - Cayley theorem, we have the following relation for $g$:
\begin{equation}
g^2 - {\rm Tr}(g) \cdot g + {\rm det}g \cdot 1 = 0.
\end{equation}
Since ${\rm det}g = 1$, we get that $g + g^{-1} = {\rm Tr}(g) \cdot 1$ and ${\rm Tr}(g) = {\rm Tr}(g^{-1})$. Also, we have the relation:
\begin{equation}
\label{cur}
\frac{{\rm Tr}^2 (g) - {\rm Tr}(g^2)}{2} = {\rm det}g = 1.
\end{equation}
We get that $p_2({\mathbb C}^1_{-})$ and $p_2({\mathbb C}^1_{+})$ are curves given by (\ref{cur}) in $U_{-}$ and $U_{+}$ respectively.
\end{proof}

\section*{Acknowledgments} The authors are grateful to A.S. Holevo for kind attention to the work and many useful remarks. The first part of the work (Sections 1, 2, 3 and 4) was fulfilled by G.G. Amosov. The second part of the work (Sections 5, 6, Appendix A and Appendix B) was
fulfilled by I.Yu. Zdanovsky. The work of G.G. Amosov is
supported by Russian Science Foundation under grant No 14-21-00162 and performed in Steklov Mathematical Institute of Russian Academy of Sciences.
The work of I.Yu. Zhanovskiy is supported by RFBR, research projects 13-01-00234 and 14-01-00416, and was prepared within the framework of
a subsidy granted to the HSE by the Government of the Russian Federation for the implementation of the Global Competitiveness Program.

\end{document}